\documentclass[a4paper,11pt]{article}
\usepackage{url}
\usepackage[utf8]{inputenc}
\usepackage{amsmath,accents,amssymb,amsthm}
\usepackage[usenames]{color}
\usepackage{longtable}
\usepackage[table]{xcolor}
\usepackage{algpseudocode}
\usepackage{amsfonts}
\usepackage{pgfplots}
\usepackage{subcaption}

\usepackage[normalem]{ulem}
\pgfplotsset{compat=1.18}

\newcommand\mut[1]{\ignorespaces}

\theoremstyle{plain}
\newtheorem{theorem}{Theorem}[section]
\newtheorem{lemma}[theorem]{Lemma}

\theoremstyle{definition}

\usepackage[mathscr]{eucal}
\usepackage[usenames]{color}

\usepackage{algorithm,algpseudocode}
\algnewcommand{\Or}{\textbf{ or }}
\algnewcommand{\And}{\textbf{ and }}
\algnewcommand{\From}{\textbf{ from }}
\algnewcommand{\To}{\textbf{ to }}
\algnewcommand{\Let}{\textbf{Let }}
\algnewcommand{\Var}{\texttt}
\algnewcommand{\Int}{\textbf{ int }}
\algnewcommand{\List}{\textbf{ list }}
\algnewcommand{\Bool}{\textbf{ bool }}
\algnewcommand{\True}{\textsc{ true }}
\algnewcommand{\False}{\textsc{ false }}
\renewcommand\leq\leqslant
\renewcommand\geq\geqslant

\usepackage{tikz,tikz-qtree,tikz-qtree-compat,adjustbox,xcolor,array,hhline}
\usepackage{pst-plot}
\usepackage{etoolbox}
\usetikzlibrary{positioning}
\providecommand\circledcolorednumb{}\renewcommand\circledcolorednumb[2]{\resizebox{0.101111\textwidth}{!}{\tikz[baseline=(char.center)]{\node[shape = circle,draw, inner sep = 2pt,fill=#1](char)    {\phantom{00}};\node[anchor=center] at (char.center) {\makebox(0,0){\large{{\sf #2}}}};}}}
\providecommand\dotscircles{}\renewcommand\dotscircles{\resizebox{0.101111\textwidth}{!}{\dots}}

\robustify{\circledcolorednumb}
\providecommand\nongap{}\renewcommand\nongap[1]{\circledcolorednumb{yellow}{#1}}

\providecommand\gap{}\renewcommand\gap[1]{\circledcolorednumb{black!20}{\phantom{#1}}}
\providecommand\generator{}\renewcommand\generator[1]{\circledcolorednumb{orange!80}{#1}}

\title{Leaf-discriminating trees for multiparameter counting of numerical semigroups}
\author{Maria Bras-Amorós and Vicenç Torra}

\begin{document}
\maketitle

\begin{abstract}
  A {\em numerical semigroup} is a cofinite additive submonoid of ${\mathbb N}$.
  The largest element of the complement is its {\em Frobenius number} and the cardinality of its complement is the {\em genus}.
  We propose a new approach to counting semigroups by the Frobenius number and by the genus, by introducing two ad-hoc trees. Those are leaf-discriminating trees in the sense that their leaves correspond exactly to the objects we want to count and, so, exploring these trees instead of general trees is optimal. Our approach reduces the search space for genus to 23\% and for Frobenius to 0.03\% of that required by previous approaches and this allowed to obtain an important number of unkown terms from both sequences.

  We implemented a recursive descending algorithm based on the so-called seeds structure, trimming the general semigroup tree exactly at those nodes with no descendants with a given genus, in the first case, or with no descendants with a given Frobenius number, in the second case.
We refined the parallelizing strategies and we overcame the previous limitation of the length of integers in the bitwise representation of the gap sequence and the seed sequence.
We computed multiparameter decomposition of the numbers
of semigroups of each Frobenius number up to $136$, the number of irreducible semigroups of Frobenius number up to $136$, and of the numbers of semigroups of each genus up to $80$.
\end{abstract}

\begin{center}
{\small {\bf Keywords:}
  Numerical semigroup, semigroup tree, genus, Frobenius number, algorithm, parallelization
}

{\small {\bf 2020 Mathematics Subject Classification:} 06F05, 20M14, 05A99, 68W30, 11Y55 }
\end{center}

\section{Introduction}
  A {\em numerical semigroup} is a cofinite additive submonoid of ${\mathbb N}$.
  The largest element of the complement is its {\em Frobenius number} and the cardinality of its complement is the {\em genus}. The least non-zero element is the {\em multiplicity}.
  
There is a large literature on counting numerical semigroups by the Frobenius number \cite{fundamentalgaps,DK,DKM} and by the genus \cite{NivaldoMedeiros, Komeda, Br:fibonacci, FromentinHivert, seeds1, rgd, seeds2, DelgadoEliahouFromentin, unleaved}.
The number of numerical semigroups of genus $g$ is denoted $n_g$ and the terms up to $n_{77}$ are known today.
The sequence can be looked up at the On-Line Encyclopedia of Integer Sequences 
({\tt https://oeis.org/A007323}) \cite{oeisA007323}.
Its terms have been generally computed by exploring the {\em semigroup tree} ${\mathscr T}$ first introduced in \cite{fundamentalgaps}, whose nodes are all numerical semigroups, whose root is ${\mathbb N}_0$, and where the parent of a non-trivial numerical semigroup $\Lambda$ with Frobenius number $F$ is the semigroup $\Lambda\cup\{F\}$.

The number of semigroups with Frobenius number $F$ will be denoted $N_F$. The terms up to $N_{39}$ can be found in \cite{fundamentalgaps}, the terms up to $N_{60}$ can be found in \cite{DKM}
and the terms up to $N_{62}$ in \cite{DK}.
It has been studied in \cite{Backelin,Li}.
Martin Fuller computed the number of semigroups of each Frobenius number up to $100$. It can be looked up at the On-Line Encyclopedia of Integer Sequences 
({\tt https://oeis.org/A124506}) \cite{oeisA124506}.

In \cite{BT:recurrence},
a computable formula is given
for the exact value $n_{F,m,g}$ of numerical semigroups with Frobenius number $F$, multiplicity $m$ and genus $g$, for the cases in which $m\geq \frac{F+1}{3}$.
The range where this formula is satisfied asymptotically covers most numerical semigroups, either of a given Frobenius number or of a given genus. This is a consequence of the known fact that as the genus grows to infinity, the Frobenius number of a random numerical semigroup of that given genus approaches twice the multiplicity of the semigroup \cite{KaplanYe,Singhal_Frobenius_genus}, and that the same behavior holds for the case when the Frobenius number grows to infinity \cite{Backelin}.
The formula in \cite{BT:recurrence} states
\begin{equation}\label{BTformula}
  n_{F,m,g}=\sum_{\{A\in B_{d}\}}\binom{\nu_{d,m}(A)}{\gamma_{d,F}(A)-g},
\end{equation}
  where $B_d=\{A\subseteq[1,d-1]\mbox{ such that }d\not\in A+A\}$, $\sigma_d(A)=\#(A\cup(A+A))\cap[1,d-1]$, $\gamma_{d,F}(A)=F-2-\#A-\sigma_d(A)$, and $\nu_{d,m}(A)=m-2-\sigma_d(A)$.

The seeds algorithm \cite{seeds1,seeds2} is at the moment the most efficient algorithm to explore the semigroup tree ${\mathscr T}$. It encodes semigroups by long integers and operates bitwise on these integers. 
It can be parallelized by splitting semigroups by their first three jumps (the first one being the multiplicity, the second one being the difference of the second nonzero nongap and the multiplicity, and the third one being the difference between the third and second nonzero nongaps). Hence, apart from being efficient, it allows multiparameter counting by considering not only the genus, but also the multiplicity and the second and third jumps.
  Unfortunately, the current implementation does not scale up for genus larger than 64.

  The terms $n_{76}$ and $n_{77}$ were obtained in~\cite{unleaved} using another algorithm, the right-generators descendant algorithm (RGD). The RGD algorithm~\cite{rgd} encodes semigroups in arrays instead of long integers which makes it scalable, although a bit slower. In~\cite{unleaved}, the RGD descending algorithm is used to explore the so-called unleaved tree of numerical semigroups up to a given genus, which is obtained by prunning ${\mathscr T}$, excluding nodes that are guaranteed to be unnecessary for the calculation. The resulting algorithm may be called the RGD-unleaved algorithm.
The idea of pruning branches of the tree ${\mathscr T}$ has previously been used in the context of the search of counterexamples of Wilf’s conjecture \cite{Wilf} (or to the related Eliahou’s conditions \cite{eliahou}) in \cite{DelgadoTrimming,DelgadoEliahouFromentin}, in which nodes whose descendants are known to verify the conjectures are not further explored.

The leaves of the general tree ${\mathscr T}$ are those semigroups that have no descendants. We will show that ${\mathscr T}$ has leaves of any genus larger than or equal to $3$ and leaves of any Frobenius number larger than or equal to $9$. Then we will
introduce the general leaf-discriminating trees as an extension of the unleaved tree used in the RGD-unleaved algorithm, now for counting semigroups either by the genus or by the Frobenius number.
The {\em genus-leaf-discriminating tree} is what was called the unleaved tree in \cite{unleaved}. It is a prunning of ${\mathscr T}$ that contains all semigroups of a given genus $g$ and has only leaves of genus equal to that $g$.
The {\em Frobenius-leaf-discriminating tree} will be a prunning of ${\mathscr T}$ that will contain only nodes with descendants with a given Frobenius number and its leaves will be exactly the semigroups of that given Frobenius number.
We will give results that characterize the genus-leaf-discriminating tree and the Frobenius-leaf discriminating tree.
  Then we will explain how the pruned trees are constructed while descending in ${\mathscr T}$ by means of deciding at each node which descendants should be in the pruned tree and which ones should not.

We adapted the seeds algorithm to explore the leaf-discriminating trees.
One advantadge of the seeds algorithm is that it descends only up to three levels above the leaves. This represents exploring only around $15\%$ of the nodes explored by the RGD-unleaved algorithm. 

We provide efficient ways to parallelize the algorithm by a detailed analysis of the abundance of the semigroups of each combination of the first three jumps. 
With this parallelization we obtained $n_{78}=76122842716469053$,
$n_{79}=123262186663081640$,
$n_{80}=199583757387067387$, as well as the number of semigroups of each Frobenius number up to $136$. Using the same algorithm we also computed the number of {\em irreducible} numerical semigroups, that is, semigroups that can not be obtained as a proper intersection of two semigroups, for Frobenius number up to $136$ (i.e., for genus up to $67$). 
In addition, we have (i) the multiparameter decomposition of all numbers $n_g$ up to $n_{80}$, by splitting by the multiplicity, and second and third jumps; (ii) the multiparameter decomposition of the number of semigroups of each Frobenius number up to $136$, by splitting by the genus and multiplicity.

  In Section~\ref{s:arbreT} we explain the general tree ${\mathscr T}$ and show that it has leaves of any possible genus $g\geq 3$ and leaves of any Frobenius number $F\geq 9$.
In Section~\ref{s:arbres} we introduce the leaf-discriminating trees.
In Section~\ref{s:algorithmg} we present a seeds-based algorithm for exploring the genus-leaf-discriminating tree and provide $n_{78},n_{79},n_{80}$.
On the other hand, in Section~\ref{s:algorithmF} we present a seeds-based algorithm for exploring the Frobenius-leaf-discriminating tree
and we give the number of numerical semigroups of each Frobenius number up to $136$.
Section~\ref{s:irred} is dedicated to the number of irreducible semigroups of each Frobenius number.
In Section~\ref{s:parallel} we explain our parallelization strategies when counting by the genus.

\section{The tree of numerical semigroups and its leaves}
\label{s:arbreT}

Given a numerical semigroup $\Lambda$, denote $F(\Lambda)$ its Frobenius number.
The {\em left elements} of $\Lambda$ are those nongaps that are smaller than $F(\Lambda)$. They are denoted $L(\Lambda)$. We will use the notation $\langle L(\Lambda)\rangle$ to denote the monoid generated by the left elements of $\Lambda$. The {\em minimal generators} of $\Lambda$ are those elements in $\Lambda\setminus\{0\}$ that cannot be expressed as the sum of two non-zero elements of $\Lambda$. It is well know that a set of integers generates a numerical semigroup if and only if they are coprime.
See for instance \cite{RGS}.
The minimal generators that are larger than $F(\Lambda)$ are called {\em right generators}. The right generators of a semigroup of Frobenius number $F$ and multiplicity $m$ lie in the interval between $F+1$ and $F+m$.

The semigroup tree, denoted ${\mathscr T}$, has as nodes all numerical semigroups, its root is ${\mathbb N}_0$, and the parent of a non-trivial numerical semigroup $\Lambda$ is the semigroup $\Lambda\cup\{F(\Lambda)\}$.
It was first introduced in \cite{fundamentalgaps}.
The children of a numerical semigroup are obtained by removing, one by one, the right generators of the parent. We will say that the child obtained by removing the smallest right generator is the {\em primogenial child}. We will also say that the {\em preceding sibling} of a non-primogenial child which is obtained by removing a right generator $\sigma$ is the child obtained by removing the largest right generator of the parent that is smaller than $\sigma$.

The leaves of the general tree ${\mathscr T}$ are those semigroups that have no right generators. The next lemmas show that there are leaves of any genus $g\geq 3$ and leaves of any Frobenius number $F\geq 9$ (and for any odd Frobenius numer $F\geq 5$). 
For two integers $a,b$, let us denote by $[a,b]$ the set $\{a,a+1,\dots,b\}$,
and let $[a,\infty)$ be the set $\{i\in{\mathbb N}_0: i\geq a\}$.

\begin{lemma}
  Suppose $F$ is an odd integer with $F\geq 5$. Then, the set $$\{0\}\cup \left[\frac{F+1}{2},F-1\right] \cup[F+1,\infty)$$ is a numerical semigroup, its Frobenius number is $F$ and it is a leaf of ${\mathscr T}$.
\end{lemma}

\begin{proof}
  It is obvious that the set is a numerical semigroup, since any sum of any two non-zero elements will be at least $F+1$. Its Frobenius number is $F$ and its multiplicity is $m=\frac{F+1}{2}$.

  All elements in the interval $[F+1,F+m]$, where right generators may occur, that is, $[F+1,F+\frac{F+1}{2}]$ is included in
  the set of sums of two elements of $\left[\frac{F+1}{2},F-1\right]$. Indeed, the set of sums is the interval between $[F+1,2F-2]$ and, since $F\geq 5$, we have $2F-2\geq F+\frac{F+1}{2}$. Hence, there are no right generators and, consequently, the semigroup is a leave in the semigroup tree ${\mathscr T}$.
\end{proof}

\begin{lemma}
  Suppose $F$ is an even integer with $F\geq 10$.
  Then, the set $$\{0\}\cup\left\{\frac{F-2}{2}\right\}\cup\left[\frac{F+4}{2},F-1\right]  \cup[F+1,\infty)$$ is a numerical semigroup, its Frobenius number is $F$ and it is a leaf of ${\mathscr T}$.
\end{lemma}

\begin{proof}
  The set is a numerical semigroup, since any sum of any two non-zero elements will be either $F-2$ or at least $F+1$.
If $F\geq 10$, then $F-2$ belongs to $\left[\frac{F+4}{2},F-1\right]$, since $F-2\geq  \frac{F+4}{2}$. The Frobenius number of the semigroup is $F$ and its multiplicity is $m=\frac{F-2}{2}$.

Let us see that all elements in the interval $[F+1,F+m]$, where right generators may occur, that is, $[F+1,F+\frac{F-2}{2}]$ is included in
the set of sums of two elements of $\left\{\frac{F-2}{2}\right\}\cup \left[\frac{F+4}{2},F-1\right]$.
Indeed, if $F\geq 10$, the interval $\left[\frac{F+4}{2},F-1\right]$ contains at least three different elements, $\frac{F+4}{2},\frac{F+6}{2},\frac{F+8}{2}$.
Then, the sums
$m+\frac{F+4}{2},m+\frac{F+6}{2},m+\frac{F+8}{2}$ cover $[F+1,F+2,F+3]$, while the sums ot two elements of the interval $\left[\frac{F+4}{2},F-1\right]$ cover $[F+4,2F-2]$ and, since $F\geq 8$, $2F-2\geq F+m=F+\frac{F-2}{2}$.

  Hence, there are no right generators and, consequently, the semigroup is a leaf in the semigroup tree ${\mathscr T}$.
\end{proof}

\begin{lemma}
  Suppose $g$ is an integer with $g\geq 3$. Then, the set $$\{0\}\cup [g,\dots,2g-2]\cup[2g,\infty)$$ is a numerical semigroup, its genus is $g$ and it is a leaf of ${\mathscr T}$.
\end{lemma}

\begin{proof}
  It is obvious that the set is a numerical semigroup, since any sum of any two non-zero elements will be at least $2g$. Its genus is $g$, its Frobenius number is $F=2g-1$, and its multiplicity is $m=g$.

  All elements in the interval $[F+1,F+m]$, where right generators may occur, that is, $[2g,3g-1]$ is included in  
  the set of sums of two elements of $\left[g,2g-2\right]$. Indeed, the set of sums is the interval between $[2g,4g-4]$ and, since $g\geq 3$, we have $4g-4\geq 3g-1$.

  Hence, there are no right generators and, consequently, the semigroup is a leaf in the semigroup tree ${\mathscr T}$.
\end{proof}

\section{Leaf-discrimintating trees}
\label{s:arbres}

Suppose that we can associate to each numerical semigroup a parameter $P$ such that it strictly increases along the semigroup tree ${\mathscr T}$.
It is the case, for instance, of the genus, the Frobenius number, or the least right generator (which we consider to be $\infty$ if it does not exist).

Given an intger $p$, which is an eventual value $P$ can take, consider the set ${\mathscr N}_{\leq p}$ of semigroups that contains all semigroups having descendants in ${\mathscr T}$ with the parameter $P$ equal to $p$.
If a semigroup $\Lambda$ is in this set, so is $\Lambda\cup\{F(\Lambda)\}$.
Then we can define the $P$-leaf-discriminating tree associated to a particular value $p$ as the tree whose nodes are all semigroups in ${\mathscr N}_{\leq p}$  and such that each non-trivial semigroup $\Lambda$ is connected to $\Lambda\cup\{F(\Lambda)\}$. If not empty, this is, indeed, a connected tree and its root is the trivial semigroup ${\mathbb N}_0$.

The leaves of the $P$-leaf-discriminating tree are then exactly all semigroups that have the parameter $P$ equal to $p$. Hence, counting semigroups by $P$ will be equivalent to counting the number of leaves of the $P$-leaf-discriminating trees.

For a fixed genus $g$, let us call ${\mathscr G}_g$ the genus-leaf-discriminating tree associated to $g$.
The next lemma is shown in \cite{unleaved}.
We use $\gcd(L(\Lambda))$
to denote the greatest common divisor of the left elements of $\Lambda$.

\begin{lemma}\label{l:charactg}
  The genus-leaf-discriminating tree ${\mathscr G}_g$ contains all nodes corresponding to semigroups $\Lambda$ of genus at most $g$ such that either
\begin{enumerate}\item $\gcd(L(\Lambda))\neq 1$,
\item $\gcd(L(\Lambda))=1$ and
$\langle L(\Lambda)\rangle$
    has genus larger than or equal to $g$.
  \end{enumerate}
\end{lemma}
The semigroups of genus $g$ are exactly the leaves of ${\mathscr G}_g$.
In Figure~\ref{fig:unleavedtreeg7} one can see ${\mathscr G}_7$.
We want to remark, that, from a computational point of view, if $\Lambda'$ is the preceding sibling of $\Lambda$ in the semigroup tree ${\mathscr T}$ or its parent in case there is no preceding sibling, then $\langle L(\Lambda)\rangle$ can be computed from $\langle L(\Lambda')\rangle$. See \cite{unleaved} for the details. This makes it possible to compute $\gcd(L(\Lambda))$ and $\langle L(\Lambda)\rangle$ at each node of the tree ${\mathscr T}$.

\begin{figure}[ht]
  \input{inputfile-unleavedg-trim-list-g7}
  \caption{Trimmed tree ${{\mathcal G}_7}$}
  \label{fig:unleavedtreeg7}
  \end{figure}~\begin{figure}[h]
%
%
%
%

\providecommand\circledcolorednumb{}\renewcommand\circledcolorednumb[2]{\resizebox{0.065741\textwidth}{!}{\tikz[baseline=(char.center)]{\node[shape = circle,draw, inner sep = 2pt,fill=#1](char)    {\phantom{00}};\node[anchor=center] at (char.center) {\makebox(0,0){\large{{\sf #2}}}};}}}
\providecommand\dotscircles{}\renewcommand\dotscircles{\resizebox{0.065741\textwidth}{!}{\dots}}
\robustify{\circledcolorednumb}
\providecommand\nongap{}\renewcommand\nongap[1]{\circledcolorednumb{yellow}{#1}}
\providecommand\gap{}\renewcommand\gap[1]{\circledcolorednumb{black!20}{\phantom{#1}}}
\providecommand\generator{}\renewcommand\generator[1]{\circledcolorednumb{orange!80}{#1}}
{\begin{tikzpicture}[grow'=right, sibling distance=6.000000mm]\tikzset{level 1/.style={level distance=7.000000cm}}\tikzset{level 2/.style={level distance=8.750000cm}}\tikzset{level 3/.style={level distance=11.375000cm}}\tikzset{level 4/.style={level distance=14.000000cm}}\tikzset{level 5/.style={level distance=17.500000cm}}\tikzset{level 6/.style={level distance=17.850000cm}}\tikzset{level 7+/.style={level distance=18.900000cm}}\node (arbre) at (current page.north) {\adjustbox{max width=\textwidth,max height=.9\textheight} { \Tree[.{\begin{tabular}{c}{$\nongap{0}\generator{1}\nongap{2}\dotscircles$} \\\end{tabular}} [.{\begin{tabular}{c}{$\nongap{0}\gap{1}\generator{2}\generator{3}\dotscircles$} \\\end{tabular}} [.{\begin{tabular}{c}{$\nongap{0}\gap{1}\gap{2}\generator{3}\generator{4}\generator{5}\dotscircles$} \\\end{tabular}} [.{\begin{tabular}{c}{$\nongap{0}\gap{1}\gap{2}\gap{3}\generator{4}\generator{5}\generator{6}\generator{7}\dotscircles$} \\\end{tabular}} [.{\begin{tabular}{c}{$\nongap{0}\gap{1}\gap{2}\gap{3}\gap{4}\generator{5}\generator{6}\generator{7}\generator{8}\generator{9}\dotscircles$} \\\end{tabular}} [.{\begin{tabular}{c}{$\nongap{0}\gap{1}\gap{2}\gap{3}\gap{4}\gap{5}\generator{6}\generator{7}\generator{8}\generator{9}\generator{10}\generator{11}\dotscircles$} \\\end{tabular}} [.{\begin{tabular}{c}{$\nongap{0}\gap{1}\gap{2}\gap{3}\gap{4}\gap{5}\gap{6}\generator{7}\generator{8}\generator{9}\generator{10}\generator{11}\generator{12}\generator{13}\dotscircles$} \\\end{tabular}} [.{\begin{tabular}{c}{$\nongap{0}\gap{1}\gap{2}\gap{3}\gap{4}\gap{5}\gap{6}\gap{7}\generator{8}\generator{9}\generator{10}\generator{11}\generator{12}\generator{13}\generator{14}\generator{15}\dotscircles$} \\\end{tabular}} ]][.{\begin{tabular}{c}{$\nongap{0}\gap{1}\gap{2}\gap{3}\gap{4}\gap{5}\nongap{6}\gap{7}\generator{8}\generator{9}\generator{10}\generator{11}\nongap{12}\generator{13}\dotscircles$} \\\end{tabular}} ]][.{\begin{tabular}{c}{$\nongap{0}\gap{1}\gap{2}\gap{3}\gap{4}\nongap{5}\gap{6}\generator{7}\generator{8}\generator{9}\nongap{10}\generator{11}\dotscircles$} \\\end{tabular}} [.{\begin{tabular}{c}{$\nongap{0}\gap{1}\gap{2}\gap{3}\gap{4}\nongap{5}\gap{6}\gap{7}\generator{8}\generator{9}\nongap{10}\generator{11}\generator{12}\nongap{13}\dotscircles$} \\\end{tabular}} ]][.{\begin{tabular}{c}{$\nongap{0}\gap{1}\gap{2}\gap{3}\gap{4}\nongap{5}\nongap{6}\gap{7}\generator{8}\generator{9}\nongap{10}\nongap{11}\dotscircles$} \\\end{tabular}} ]][.{\begin{tabular}{c}{$\nongap{0}\gap{1}\gap{2}\gap{3}\nongap{4}\gap{5}\generator{6}\generator{7}\nongap{8}\generator{9}\dotscircles$} \\\end{tabular}} [.{\begin{tabular}{c}{$\nongap{0}\gap{1}\gap{2}\gap{3}\nongap{4}\gap{5}\gap{6}\generator{7}\nongap{8}\generator{9}\generator{10}\nongap{11}\dotscircles$} \\\end{tabular}} [.{\begin{tabular}{c}{$\nongap{0}\gap{1}\gap{2}\gap{3}\nongap{4}\gap{5}\gap{6}\gap{7}\nongap{8}\generator{9}\generator{10}\generator{11}\nongap{12}\nongap{13}\dotscircles$} \\\end{tabular}} ]][.{\begin{tabular}{c}{$\nongap{0}\gap{1}\gap{2}\gap{3}\nongap{4}\gap{5}\nongap{6}\gap{7}\nongap{8}\generator{9}\nongap{10}\generator{11}\dotscircles$} \\\end{tabular}} ]][.{\begin{tabular}{c}{$\nongap{0}\gap{1}\gap{2}\gap{3}\nongap{4}\nongap{5}\gap{6}\generator{7}\nongap{8}\nongap{9}\dotscircles$} \\\end{tabular}} [.{\begin{tabular}{c}{$\nongap{0}\gap{1}\gap{2}\gap{3}\nongap{4}\nongap{5}\gap{6}\gap{7}\nongap{8}\nongap{9}\nongap{10}\generator{11}\dotscircles$} \\\end{tabular}} ]][.{\begin{tabular}{c}{$\nongap{0}\gap{1}\gap{2}\gap{3}\nongap{4}\nongap{5}\nongap{6}\gap{7}\nongap{8}\nongap{9}\dotscircles$} \\\end{tabular}} ]][.{\begin{tabular}{c}{$\nongap{0}\gap{1}\gap{2}\nongap{3}\gap{4}\generator{5}\nongap{6}\generator{7}\dotscircles$} \\\end{tabular}} [.{\begin{tabular}{c}{$\nongap{0}\gap{1}\gap{2}\nongap{3}\gap{4}\gap{5}\nongap{6}\generator{7}\generator{8}\nongap{9}\dotscircles$} \\\end{tabular}} [.{\begin{tabular}{c}{$\nongap{0}\gap{1}\gap{2}\nongap{3}\gap{4}\gap{5}\nongap{6}\gap{7}\generator{8}\nongap{9}\generator{10}\nongap{11}\dotscircles$} \\\end{tabular}} ]][.{\begin{tabular}{c}{$\nongap{0}\gap{1}\gap{2}\nongap{3}\gap{4}\nongap{5}\nongap{6}\gap{7}\nongap{8}\nongap{9}\dotscircles$} \\\end{tabular}} ]]][.{\begin{tabular}{c}{$\nongap{0}\gap{1}\nongap{2}\gap{3}\nongap{4}\generator{5}\dotscircles$} \\\end{tabular}} [.{\begin{tabular}{c}{$\nongap{0}\gap{1}\nongap{2}\gap{3}\nongap{4}\gap{5}\nongap{6}\generator{7}\dotscircles$} \\\end{tabular}} [.{\begin{tabular}{c}{$\nongap{0}\gap{1}\nongap{2}\gap{3}\nongap{4}\gap{5}\nongap{6}\gap{7}\nongap{8}\generator{9}\dotscircles$} \\\end{tabular}} ]]]]] }  };
\end{tikzpicture}}
  \caption{Trimmed tree ${{\mathcal F}_7}$}
  \label{fig:unleavedtreeF7}
\end{figure}
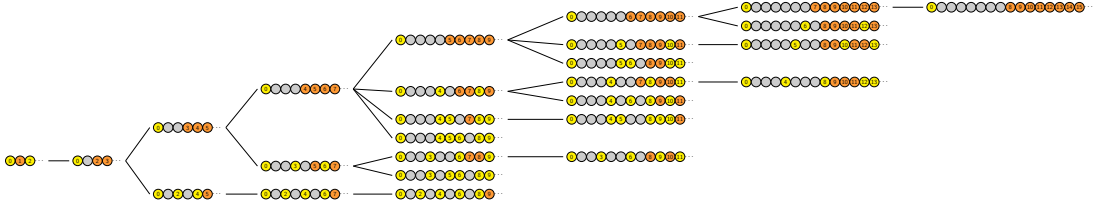

In Figure~\ref{fig:unleavedtreeg10} one can appreciate which are the nodes belonging to ${\mathscr G}_{10}$ (in blue) compared to the set of all nodes in the general tree ${\mathscr T}$ of genus at most $10$ (in both gray and blue). In this case there are exactly 364 blue nodes out of a total of $478$ nodes. The unleaved-RGD algorithm in \cite{unleaved} explores ${\mathscr G}_g$. It does not descend to the leaves but it stops at the previous layer. Hence, the number of explored nodes descends to $160$. The algorithm we present in Section~\ref{s:algorithmF} descends only up to three levels above the leaves. This way, in the example of genus $10$, it represents exploring only $45$ nodes compared to the $160$ nodes explored by the unleaved-RGD algorithm.

\begin{figure}[ht]
\input{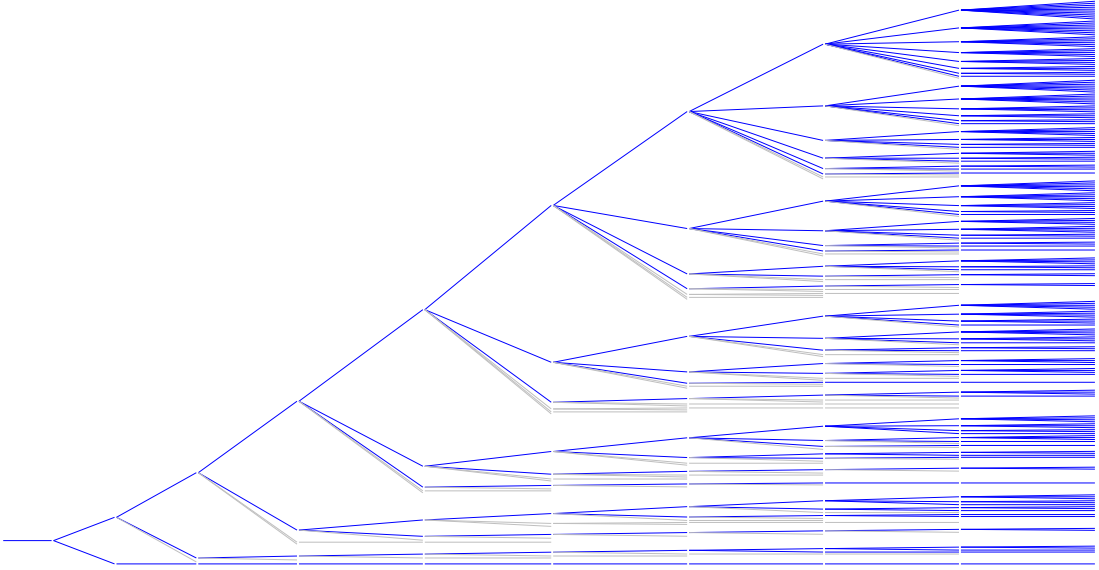}
\caption{Nodes of ${{\mathcal G}_{10}}$ (in blue) with respect to all nodes of semigroups of genus up to $10$ (in both gray and blue)}
\label{fig:unleavedtreeg10}
\end{figure}~\begin{figure}[h]
\input{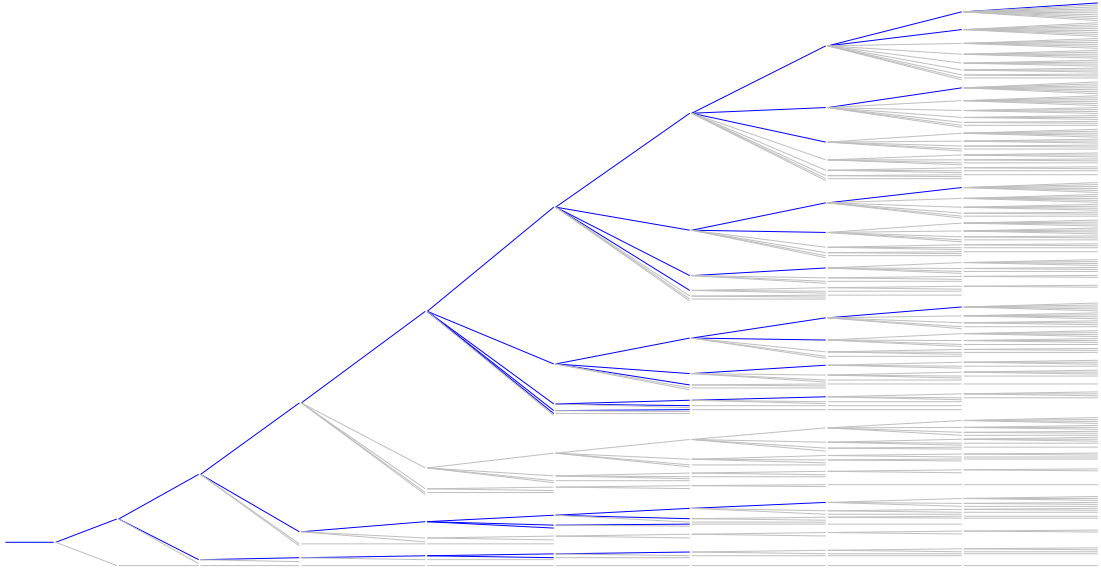}
\caption{Nodes of ${{\mathcal F}_{10}}$ (in blue) with respect to all nodes of semigroups of genus up to $10$ (in both gray and blue)}
\label{fig:unleavedtreeF10}
\end{figure}

Now, for a fixed Frobenius number $F$, denote by ${\mathscr F}_F$ the Frobenius-leaf-discriminating tree associated to $F$.
The next lemma characterizes the nodes of ${\mathscr F}_F$.

\begin{lemma}\label{lemmakeepgoingF}
  The Frobenius-leaf-discriminating tree ${\mathscr F}_F$ contains all nodes corresponding to semigroups $\Lambda$
  such that both
\begin{enumerate}
\item $F(\Lambda)\leq F$,
\item $F\not\in \langle L(\Lambda)\rangle$.
\end{enumerate}
\end{lemma}

\begin{proof}
  If $F\in \langle L(\Lambda)\rangle$, then any descendant of $\Lambda$ contains the left elements of $\Lambda$ and, so, it contains $\langle L(\Lambda)\rangle$. In particular, $F$ can not be a gap of any descendant of $\Lambda$, and so, it can not be the Frobenius number of any descendant of $\Lambda$.

  If $F\not\in\langle L(\Lambda)\rangle$, but the Frobenius number of $\Lambda$ is larger than $F(\Lambda)>F$, then, any descendant will have Frobenius number strictly larger than $F$, and, so, $F$ can not be the Frobenius number of any descendant of $\Lambda$, either.

  On the other hand, if $F\not\in \langle L(\Lambda)\rangle$ and
  the Frobenius number of $\Lambda$ is at most $F$, then $\langle L(\Lambda)\rangle \cup [F+1,\infty)$ is either $\Lambda$ or a descendant of $\Lambda$ with Frobenius number equal to $F$.
  \end{proof}

In Figure~\ref{fig:unleavedtreeF7} we draw ${\mathscr F}_{7}$. When the seeds algorithm is adapted to count semigroups by Frobenius number, it is restricted to explore ${\mathcal F}_F$.
In Figure~\ref{fig:unleavedtreeF10} one can appreciate which are the nodes in ${\mathcal F}_{10}$ (in blue) compared to the
set of all nodes in ${\mathscr T}$ of genus up to $10$. Notice that we limit the tree to genus $10$ since any semigroup with Frobenius number $10$ has genus at most $10$. In this case, there are $52$ blue nodes out of a total of $478$ general nodes. 
The trimming becomes even more dramatic when we restrict the exploration to multiplicities $m<\frac{F+1}{3}$, which can be done since for all the other multiplicities we can use the formula~\eqref{BTformula} for $n_{F,m,g}$.
In the example of Frobenius number $10$, this represents exploring only $9$ nodes.

\section{The seeds algorithm exploring the genus-leaf-dis\-cri\-mi\-na\-ting tree}
\label{s:algorithmg}

In \cite{unleaved} we explain how to explore the genus-leaf-discriminating tree. This is done thanks to the characterization in Lemma~\ref{l:charactg}.
At each node we compute the greatest common divisor (gcd) of the left elements of the corresponding semigroup. If it is not one, then we keep descending. If it is one, we only descend if the genus of the semigroup generated by the left elements is at least the target genus. The genus of the left elements, as well as the semigroup generated by them (divided by their gcd if not one) can be easily updated every time we descend from one node to its primogenial child and then, when we explore any other child, it can be easily updated from their preceding sibling. Using this for exploring just the nodes in ${\mathscr G}_g$ results in a very significant improvement of the required computation time with respect to exploring the whole semigroup tree ${\mathscr T}$.

The seeds algorithm presented in \cite{seeds1} is based on the so-called seeds of a numerical semigroup.
Given a numerical semigroup $\Lambda=\{\lambda_0=0<\lambda_1<\dots\}$ with Frobenius number $F$, a nongap $\lambda_t\in\Lambda$ with $\lambda_t>F$ is an {\em order-$i$ seed} of $\Lambda$ if
$\lambda_t+\lambda_i\neq\lambda_j+\lambda_{j'}$ for all $i<j,\,j'<t$.
In particular, the order-zero seeds of $\Lambda$ are its right generators.
One can prove that there are no seeds of order larger than $F-g$, and that order-$i$ seeds are at most $F+\lambda_{i+1}-\lambda_i$.
This allows to encode the seeds of $\Lambda$ in a binary vector of length $\lambda_1+(\lambda_2-\lambda_1)+(\lambda_3-\lambda_2)+\dots+(\lambda_{F-g+1}-\lambda_{F-g})=\lambda_{F-g+1}=F+1$. We denote this binary vector the seeds bitstream of $\Lambda$.

The seeds algorithm 
computes for each explored node its seeds bitstream as well as its gap sequence.
The set of descendants of each node is then very easily described by the seeds bitstream and there is a fast descending algorithm that computes the seeds bitstream and the gap sequence of a descendant from those of its preceding sibling in the semigroup tree or from its parent in case there is no preceding sibling.
The seeds algorithm has another remarkable advantage, introduced in
\cite{seeds2}, which is that, while the RGD algorithm descends the tree up to the level above the leaves, with the seeds algorithm we can use that we know the number of great-grand children of a node from its seeds bitstream without the need to descend the node.

At the moment, the seeds algorithm is the most efficient algorithm to explore the semigroup tree. However, the fastest available implementation is in C++ and, with the current native integers, it can only be used to explore semigroups of genus up to 64. It can be parallelized by splitting semigroups by their first three jumps.
In contrast, the largest known term of the sequence $n_g$ of semigroups of each given genus $g$ is today $n_{77}$, which has been obtained using the RGD-unleaved algorithm~\cite{rgd},
by means of exploring the genus-leaf-discriminating tree, as explained above, and descending it up to level $g-1$. The implementation in \cite{rgd} computes $n_g$ using some other tricks to visit less nodes, but disabling it to obtain the split count by the first three jumps.

We boost the performance of the seeds algorithm by adapting it to explore the genus-leaf-discriminating tree. For this we added the descending discrimination criteria at each node, and so we introduced the computation of the gcd and the genus of the semigroup generated by the left elements at each step, whenever the left elements are coprime. We perform it by specific updating strategies.

In Table~\ref{t:nodesvisitatsg} we give, for several genera $g$ up to $40$: (i) the number of nodes in ${\mathscr T}$ of genus at most $g$, (ii) the number of nodes of ${\mathcal G}_g$, (iii) the number of nodes of ${\mathcal G}_g$ up to level $g-1$, (iv) the number of nodes explored by the RGD algorithm in \cite{unleaved}, and (v) the number of nodes explored by the new seeds algorithm (seeds-LDT). We derive that the search
space of the seeds-LDT algorithm is 23\% of that required by the previous RGD-unleaved algorithm.

\begin{table}
  \begin{center}
\resizebox{\textwidth}{!}{\begin{tabular}{|r|r|r|r|r|r|}\hline
  $g$ & \begin{tabular}{l}nodes of \\genus$\leq g$\end{tabular} & \begin{tabular}{l}nodes in \\${\mathcal G}_g$\end{tabular} & \begin{tabular}{l}nodes \\in ${\mathcal G}_g$\\ of genus\\$\leq g-1$\end{tabular} & \begin{tabular}{l}nodes\\ explored\\ by RGD-unleaved\end{tabular} & \begin{tabular}{l}nodes \\explored \\by the \\seeds-LDT alg.\end{tabular} 
  \\\hline
  10 & 478 & 364 & 160 & 61 & 45 \\
  15 & 6964 & 4833 & 1976 & 1325 & 428 \\
  20 & 93142 & 61469 & 24073 & 16774 & 4312 \\
  25 & 1179597 & 759972 & 292748 & 196433 & 47278 \\
  30 & 14396338 & 9146174 & 3499401 & 2282567 & 536246 \\
  35 & 171202690 & 107815637 & 41128436 & 26454236 & 6128094 \\
  40 & 1998799015 & 1251716100 & 477101816 & 304794995 & 70187790 \\
\hline
\end{tabular}}
\end{center}
  \caption{(i) The number of nodes in ${\mathscr T}$ of genus at most $g$,
 (ii) the number of nodes of ${\mathcal G}_g$, (iii) the number of nodes of ${\mathcal G}_g$ up to level $g-1$, (iv) the number of nodes explored by the RGD-unleaved algorithm, and (v) the number of nodes explored by the new seeds-LDT algorithm.}
\label{t:nodesvisitatsg}
\end{table}

In Appendix A we include a pseudocode corresponding to the recursive descending function from a parent to its children in ${\mathcal G}_g$. In variable $S$ we keep the seeds vector of the semigroup being handled and in $spanS$ the monoid generated by its left elements.
The variable $gd$ counts the difference of the target genus and the genus of the semigroup being handled minus 3.

The base case is when $gd=0$. In this case the function returns the result of the formula of the number of great grandchildren of a semigroup from its seeds bitstream given in \cite[Theorem 9]{seeds2}. We use $w_0^a(A)$ to denote the number of $1$-bits in $A$ at positions from $0$ to $a$.

For the general case, the condition for recursion is
that the gcd of the left elements is different than one, or that, in case it is one, the genus of $spanS$ is larger than the target genus.
In the pseudocode there are three blocks.
There is a block for processing the bit at position $0$ of $S$. If its value is $1$, then there is a child with the same left elements as its parent and the conductor is the parent's conductor augmented by one. Neither the gcd of the left elements nor $spanS$  need to be updated in this case.
There is another block for processing the bit at position $1$ of $S$. If its value is $1$, then there is a child with left elements equal to the left elements of the parent together with its parents's conductor. The gcd of the new left elements is then just the gcd of two values, the first one being the gcd of the left elements of the parent, which is already computed, and the second one being the parent's conductor. The monoid $spanS$ needs to be updated in this case, by adding the parent's conductor as a new generator.
Finally there is a block for processing the rest of the bits of $S$, as far as there are left generators to be processed. In this case, if the value of a bit is $1$, then there is a child with left elements equal to the left elements of the parent together with an interval starting at the parent's conductor. The gcd of the new left elements is $1$, since the interval has at least two consecutive elements. The monoid $spanS$ needs to be updated in this case. The update of $spanS$ can be done from the monoid of the preceding sibling, should it exist.

In order to be able to explore the tree at depths larger than 64,
we implemented two independent libraries: one for 192-bits and one for 256-bits integers.

As a result of our computations we obtained
$$n_{78}=76122842716469053$$
$$n_{79}=123262186663081640$$
$$n_{80}=199583757387067387$$
Here we want to remark that $n_{50}$ was first computed by brute approach in \cite{Br:fibonacci} during 19 days. Fromentin and Hivert computed the same number with $290$ seconds in \cite{FromentinHivert}.
With {\tt posix} implementations in an ordinary laptop we currently obtain the same number in 213 seconds using the RGD algorithm, 210 seconds using the seeds algorithm, 121 seconds using the RGD-unleaved algorithm, and 40 seconds using our new implementation for exploring the genus-leaf-discriminating tree with the seeds algorithm.

Experiments suggest that the execution time grows exponentially with the genus by a factor of $1.618$.
In \cite{unleaved} $n_{77}$ was computed during 48.5 days.
With the new implementations run in Kebnekaise at HPC2N
and Tetralith at NAISS,
$n_{78}$ was obtained in 6 days, $n_{79}$ in $13$ days, and $n_{80}$ also in 13 days, because duration depends also on process priorities in the shared high performance cluster.
In core-hours, the time required to compute $n_{80}$ was $231.3 \cdot 1000$ core-hours, which corresponds to more than $26$ years.
In Section~\ref{s:parallel} we explain how we could cut the whole computation in parts that were computable within 7 days each, as required by these high performance computers.

\section{The seeds algorithm exploring the Frobenius-leaf-dis\-cri\-mi\-na\-ting tree}
\label{s:algorithmF}
Let $N_F$ denote the number of numerical semigroups with Frobenius number $F$. We present an adaptation of the seeds algorithm to compute $N_F$ that combines (i) the use of formula \eqref{BTformula}, from \cite{BT:recurrence}, for multiplicities $m$ such that $m\geq \frac{F+1}{2}$, and (ii) the exploration of the Frobenius-leaf-discriminating tree ${\mathcal F}_F$ for the other multiplicities. 
  The exploration of ${\mathcal F}_F$ exploits Lemma~\ref{lemmakeepgoingF}. With this algorithm we got the terms of the sequence $N_F$ for $F$ from $101$ to $136$.

The pseudocode of the algorithm that recursively descends from parents to their children in ${\mathcal F}_F$ is provided in Appendix B. It is similar to that of when counting by the genus. Now the variables $gd, u, v$ as well as the $gcd$ are not needed. Its general structure is as follows: 

The base case is when the semigroup being handled has Frobenius number equal to the target Frobenius number, that is when the conductor of the semigroup is the target Frobenius number plus one.

For the general case, the function has the same three blocks as before.
The main difference is that the condition for recursion is now, according to
Lemma~\ref{lemmakeepgoingF},
that the target Frobenius number is a gap of the monoid $spanS$, generated by the left elements of the semigroup being handled.

In Table~\ref{t:nodesvisitatsF} we give, for several Frobenius numbers $F$ up to $40$: (i) the number of nodes in ${\mathscr T}$ of genus at most $F$, (ii) the number of nodes of ${\mathcal F}_F$, and (iii) the number of nodes explored by our seeds-LDT algorithm, that is, nodes that are at tree depth at most $F-3$ and for which the multiplicity satisfies $m<\frac{F+1}{3}$.
In this case the search space of the seeds-LDT algorithm reduces to 0.03\% of that required by the previous algorithms exploring the whole semigroup tree ${\mathscr T}$.

\begin{table}
  \begin{center}
\begin{tabular}{|r|r|r|r|}\hline
  $F$ & nodes of genus$\leq F$ & nodes in ${\mathcal F}_F$ & explored nodes\\\hline
 10 & 478 & 52 &  9\\
 15 & 6964 & 423 &  42\\
 20 & 93142 & 1853 &  130\\
 25 & 1179597 & 16830 & 1761\\
 30 & 14396338 & 63994 & 4305\\
 35 & 171202690 & 840171 & 95024\\
40 & 1998799015 & 4601656 &  537338\\\hline
\end{tabular}
\end{center}
  \caption{(i)
The number of nodes in ${\mathscr T}$ of genus at most $F$, (ii) the number of nodes of ${\mathcal F}_F$, and (iii) the number of nodes explored by our algorithm.}
\label{t:nodesvisitatsF}
\end{table}

In Table~\ref{taulaNF} we give the number of semigroups with Frobenius number between $101$ 
and $136$. 
Computing $N_{134}$ took $253$ hours while $N_{136}$ took $102$ hours. 
We have also used this algorithm to compute the value $n_{F,m}$ of the number of numerical semigroups with Frobenius number $F$ and multiplicity $m$ for $F$ up to $136$.

\begin{table}
\resizebox{\textwidth}{!}{\begin{tabular}{|r|r||r|r||r|r|}
 \hline
$F$ & $N_F$ & $F$ & $N_F$ & $F$ & $N_F$\\
  \hline
  101 & 2859805177566292  & 113 & 183434646614916731  & 125 & 11754175518744959986 \\
  102 & 2815711168628566  & 114 & 180884445758448980  & 126 & 11601801163893032006 \\
  103 & 5721342683868839  & 115 & 366926717014466133  & 127 & 23510410318215209870 \\
  104 & 5653084347812039  & 116 & 362597785059240733  & 128 & 23234609826654274797 \\
  105 & 11419546095451295 & 117 & 732982650574525666  & 129 & 46986766251335527340 \\
  106 & 11309173804469087 & 118 & 725290053632625924  & 130 & 46471844310284935850 \\
  107 & 22907719344734948 & 119 & 1468507257142735198 & 131 & 94070406446689407239 \\
  108 & 22573332850270825 & 120 & 1448820184698501073 & 132 & 92876546429911575576\\
  109 & 45825142465629930 & 121 & 2937342548770115212 & 133 & 188151414140633928693 \\
  110 & 45285644532433944 & 122 & 2902926372048783514 & 134 & 185951270355038992184 \\
  111 & 91508738209042815 & 123 & 5869220828013043596 & 135 & 376098229544177638814 \\
  112 & 90573608825024748  & 124 & 5805952197312285718 & 136 & 371906812715845519533\\
\hline
\end{tabular}}
\caption{Number of semigroups with Frobenius number $F$, for $F$ from $101$ to $136$.}
\label{taulaNF}
\end{table}

\section{Counting irreducible numerical semigroups}
\label{s:irred}

It is well known that the Frobenius number $F$ of a numerical semigroup of genus $g$ is upper bounded by $2g-1$. This follows by the pigeon's principle, since for each pair of integers adding up to $F$ at least one of them must be a gap.
A numerical semigroup with genus $g$ and Frobenius number $F$ is said to be {\em symmetric} if $F$ is exactly $2g-1$ and {\em pseudo-symmetric} if $F=2g-2$.

The intersection of two different numerical semigroups is a numerical semigroup and the semigroups that can not be expressed as a proper intersection of two numerical semigroups are called {\em irreducible} numerical semigroups.
Branco and Rosales proved that the set of irreducible numerical semigroups is the union of the set of symmetric semigroups and the set of pseudo-symmetric semigroups \cite{BrancoRosales}.

The number of irreducible numerical semigroups of each Frobenius number was computed in \cite{fundamentalgaps} for Frobenius number up to $39$ and is published at the entry A158206 of the Online Encyclopedia of Integer Sequences
({\tt https://oeis.org/A158206}) \cite{oeisA158206}.
With our algorithm for exploring ${\mathcal F}_F$ we could compute this sequence for Frobenius number up to $136$ as shown in Table~\ref{t:irred}.

\begin{table}
 \begin{tabular}{|ccc||ccc||ccc|}
 \hline
 $F$ & $g$ & n. irred. sem. & $F$ & $g$ & n. irred. sem. & $F$ & $g$ & n. irred
sem.\\
 \hline
 38 & 20 & 182 &  71 & 36 & 19812  & 104 & 53 & 721159 \\
 39 & 20 & 227  & 72 & 37 & 9667   & 105 & 53 & 947145 \\
 40 & 21 & 196  & 73 & 37 & 25405  & 106 & 54 & 984242 \\
 41 & 21 & 420  & 74 & 38 & 19020  & 107 & 54 & 1655267 \\
 42 & 22 & 203  & 75 & 38 & 23297  & 108 & 55 & 839515 \\
 43 & 22 & 546  & 76 & 39 & 21564  & 109 & 55 & 2106583 \\
 44 & 23 & 342  & 77 & 39 & 41642  & 110 & 56 & 1570954 \\
 45 & 23 & 498  & 78 & 40 & 22178  & 111 & 56 & 2001431 \\
 46 & 24 & 527  & 79 & 40 & 53629  & 112 & 57 & 1907837 \\
 47 & 24 & 926  & 80 & 41 & 35886  & 113 & 57 & 3417576 \\
 48 & 25 & 411  & 81 & 41 & 51367  & 114 & 58 & 1904303 \\
 49 & 25 & 1182 & 82 & 42 & 51572  & 115 & 58 & 4273853 \\
 50 & 26 & 844  & 83 & 42 & 88093  & 116 & 59 & 3035371 \\
 51 & 26 & 1121 & 84 & 43 & 41858  & 117 & 59 & 4162984 \\
 52 & 27 & 981  & 85 & 43 & 109693 & 118 & 60 & 4213006 \\
 53 & 27 & 2015 & 86 & 44 & 84770  & 119 & 60 & 7023282 \\
 54 & 28 & 1039 & 87 & 44 & 106526 & 120 & 61 & 3669716 \\
 55 & 28 & 2496 & 88 & 45 & 100439 & 121 & 61 & 8945931 \\
 56 & 29 & 1715 & 89 & 45 & 184466 & 122 & 62 & 6829161 \\
 57 & 29 & 2436 & 90 & 46 & 98332  & 123 & 62 & 8512341 \\
 58 & 30 & 2499 & 91 & 46 & 233557 & 124 & 63 & 8001407 \\
 59 & 30 & 4350 & 92 & 47 & 159833 & 125 & 63 & 14243981 \\
 60 & 31 & 1857 & 93 & 47 & 222990 & 126 & 64 & 8191764  \\
 61 & 31 & 5602 & 94 & 48 & 227059 & 127 & 64 & 18354729 \\
 62 & 32 & 4173 & 95 & 48 & 375582 & 128 & 65 & 13300505 \\
 63 & 32 & 5317 & 96 & 49 & 195266 & 129 & 65 & 17488308 \\
 64 & 33 & 4866 & 97 & 49 & 490585 & 130 & 66 & 17564858 \\
 65 & 33 & 8925 & 98 & 50 & 369450 & 131 & 66 & 29633619 \\
 66 & 34 & 4839 & 99 & 50 & 471036 & 132 & 67 & 15515628 \\
 67 & 34 & 11971&100 & 51 & 419865 & 133 & 67 & 37519331 \\
 68 & 35 & 7826 &101 & 51 & 799237 & 134 & 68 & 28837942 \\
 69 & 35 & 11276&102 & 52 & 439443 & 135 & 68 & 35684844 \\
 70 & 36 & 10977&103 & 52 & 1018271& 136 & 69 & 34851905 \\
 \hline
\end{tabular}

\caption{Number of irreducible semigroups of each Frobenius number from $38$ to
$136$.}
\label{t:irred}
\end{table}

\section{Parallelization of the seeds algorithm for counting by the genus}
\label{s:parallel}
Let $m$ be the multiplicity of a numerical semigroup
and let $u$ and $v$ be the second and third jumps, respectively. That is, $u$ is the difference between the second largest nonzero nongap and the multiplicity and $v$ is the difference between the third and the second largest nonzero nongaps. 
Let $w_g(m,u,v)$ be the number of semigroups of genus $g$, multiplicity $m$, second jump $u$ and third jump $v$.

Roughly speaking, our program generates a parallelizable task for each combination of parameters $m$, $u$, $v$ (except for some extremal values such as for $m=2$ or $m=g+1$). That is, given parameters $m,u,v$ we compute in parallel $w_g(m,u,v)$. The computation time to obtain $w_g(m,u,v)$ is proportional to $w_g(m,u,v)$ itself. For a good performance, and due to high performance computing requirements, it is desirable that computational demands of different tasks are balanced. Because of that, computation is further subdivided for selected $m, u, v$ with very large and time-consuming $w_g(m,u,v)$. In contrast, the least time consuming $w_g(m,u,v)$ are grouped in single tasks. 
For these two reasons we need an estimated order of the tasks, from the most time-consuming to the least one.

Each task corresponding to the parameters $m,u,v$ is prioritized in terms of $m$ and $u+v$. This means that we give the same priority to all tasks with the same multiplicity $m$ and for which the sum $u+v$ is the same, which implies that $m$ and $m+u+v$, i.e. the first and third nonzero nongaps of the computed semigroups, are the same. For this reason we define
$w_g(m,\omega)=\sum_{u+v=\omega} w_g(m,u,v)$
and
we aim to order the tasks by decreasing order of $\frac{w_g(m,u+v)}{u+v-1}$, as this is the mean value of the number of elements in $w_g(m,u+v)$ over all combinations of $m,u,v$ with the same sum $u+v$.
This is equivalent to ordering the tasks by decreasing order of
$\log_2\left(\frac{w_g(m,u+v)}{u+v-1}\right).$

Consider the matrix with the values
$w_g(m,\omega)$ with $2\leq m\leq g+1$ and $2\leq \omega\leq \lfloor\frac{2(g+3)}{3}\rfloor$. The upper bound on $\omega$ is justified in Lemma~\ref{omegabound} below.
The columns of this matrix have maximum values over the line $m+\omega=g\gamma+2\varphi$.
Indeed, for
a numerical semigroup $S=\{\lambda_0=0,\lambda_1,\lambda_2\dots\}$ with $\lambda_i<\lambda_{i+1}$, of genus $g$, define the following map $\psi_S:[0,1]\rightarrow[0,2]$, such that $\psi_S(\alpha)=\frac{\lambda_{1+\lfloor\alpha(g-1)\rfloor}}{g}.$
Define also
$$\psi(\alpha)=\left\{\begin{array}{ll}\gamma+\varphi\alpha&\mbox{ if }0\leq\alpha\leq\frac{1}{\sqrt{5}},\\1+\alpha&\mbox{ if }\frac{1}{\sqrt{5}}\leq \alpha\leq 1,
\end{array}\right.$$
where $\varphi=\frac{1+\sqrt{5}}{2}$ and $\gamma=\frac{5+\sqrt{5}}{10}$.
It is proved in \cite{BrasKaplanSinghal} that, as $g$ grows to infinity, $\psi_S$ approaches $\psi$ uniformly in probability.
Let $t$ be an integer such that $t\leq 1+\frac{g-1}{\sqrt{5}}$.
The $t$th nonzero element of a semigroup $S$ of genus $g$ is $\lambda_t=g\psi_S\left(\frac{t-1}{g-1}\right),$ and, by the aforementioned result, this is expected to be very close to $g\psi\left(\frac{t-1}{g-1}\right)=g\gamma+(t-1)\varphi\frac{g}{g-1}.$

Since the value $m+\omega$ equals the third nonzero nongap,
by the previous remark, it should be around $g\gamma+2\varphi$. 
Hence, the columns in the matrix have maximum values around the line $m+\omega=\rho(g)$, where 
$\rho(g)$ is the integer rounding of $g\gamma+2\varphi$.

Dividing each column by $\omega-1$, and also taking the logs of all the obtained entries, keeps the property that the columns in the matrix have maximum values over the line $m+\omega=\rho(g)$.
For this reason, we chose to approximate the matrix of
$\log_2\left(\frac{w_g(m,\omega)}{\omega-1}\right)$ by a double plane with edge in the line $m+\omega=\rho(g)$.
The double plane of a target genus $t$ can be estimated by the values of a previous reference genus $r$. 
It is known \cite{Zhai} that the number of numerical semigroups of genus $t$, $n_t$, can be approximated by $\varphi^{t-r} n_r$.
For this reason we work with the hypothesis that,
for integer values of $m$, $u$ and $v$ with $m+u+v$ close to $\rho(t)$,
$w_t(m,u,v)\approx \varphi^{t-r}w_r(m,u,v),$
and from this approximated values we can assign a priority to each pair $(m,\omega)$.
Once we have this task ordering, tasks are grouped into batches which are submitted independently to the high performance computer Kebnekaise at HPC2N. Each batch requests
from 8 to 32 cores depending on the estimation of time consumption.
In addition, the computation also establishes whether the number of semigroups of $(m,u,v)$ is computed in a single run, or if the computation is split in 4, 8, 16, or 32 subtasks. 

We conclude this section with a bound on the sum $u+v$, which is needed to bound the dimension of the prioritizing matrix.

\begin{lemma}\label{omegabound}
  If there exists a numerical semigroup of genus $g$ with first jumps $m,u,v$, then
  $u+v\leq\lfloor\frac{2(g+3)}{3}\rfloor$.
  \end{lemma}
\begin{proof}
  On one hand we have $g\geq m+u+v-3$ and on the other hand we have $u+v\leq 2m$.
  Hence,
  $$u+v\leq\min(g-m+3,2m).$$

  The result is proved by considering two cases.

  \begin{itemize}
  \item If $m\leq \frac{g+3}{3}$ then 
    $u+v\leq\min(g-m+3,2m)=2m\leq \frac{2(g+3)}{3}$.

  \item If $m\geq \frac{g+3}{3}$ then
    $u+v\leq\min(g-m+3,2m)=g-m+3\leq g+3-\frac{g+3}{3}=\frac{2(g+3)}{3}$.
  \end{itemize}

  Hence, $u+v\leq \frac{2(g+3)}{3}$ and, since $u+v$ is an integer, we conclude that   $u+v\leq \lfloor\frac{2(g+3)}{3}\rfloor$.

\end{proof}

\section{Conclusion}

We presented new approaches for making computations in the semigroup tree and obtain results never obtained before in much less computational time. One main tool is controlling the leafs and defining new search trees.

We presented a new version of the seeds algorithm for counting semigroups by genus, but also for giving the number of semigroups of each genus and each combination of the three first jumps of a semigroup, the first one being the multiplicity. We used in-depth parallelization tools and long integer bitwise computation that allowed us to get to $n_{80}$.

We presented and described the Frobenius-leaf-discriminating tree and used it, together with the former formulas, to compute $N_F$ and the number of irreducible semigroups of Frobenius number $F$ for $F$ up to $136$. We also computed the multiparameter decomposition of these numbers with respect to the multiplicity and the genus.

Our approach reduces the search space for genus to 23\% and for
    Frobenius to 0.03\% of that required by previous approaches.
We contributed to four sequences of the Online Encyclopedia of Integer Sequences:
\begin{itemize}
\item {\tt A007323: Number of numerical semigroups of genus n} \cite{oeisA007323}, \\terms {\tt a(n) for n=78..80},
\item {\tt A124506: Number of numerical semigroups with Frobenius number n} \cite{oeisA124506},
  \\terms {\tt a(n) for n=101..136},
\item {\tt A158206: Number of irreducible numerical semigroups with Frobenius number n} \cite{oeisA158206}, \\terms {\tt a(n) for n=40..136},
  \item {\tt A398185: Triangular array: T(n,m) gives the number of numerical semigroups of Frobenius number n and multiplicity k, (n>=1, 2<=k<=n+1)} \cite{oeisA398185}, \\terms {\tt a(n) for n=1..5050}.
\end{itemize}

\section*{Acknowledgement}

This research is supported by project PID2024-156636NB-C21 (MATSE), funded by MCIN/AEI/10.13039/501100011033/ FEDER, UE.

This work was partially supported by the Wallenberg AI, Autonomous Systems and
Software Program (WASP) funded by the Knut and Alice Wallenberg Foundation.
This research was conducted using the resources of High Performance Computing
Center North (HPC2N).
The computations were enabled by resources provided by the National Academic
Infrastructure for Supercomputing in Sweden (NAISS), partially funded by the
Swedish Research Council through grants agreement no. 2025-22-1579, NAISS
2026/3-170.

The authors would like to thank Axel Bacher and Roman Iakymchuk for their fruitful comments.


\section*{Appendix A}
\begin{algorithmic}[1]
\Function{Sdesc$_g$}{$S, \text{\textit{spanS}}, c, m, u, v, r, \text{\textit{gd}}, \text{\textit{gcd}}$}
     \State $ng \leftarrow 0$
     \vspace{1mm}
     \If{$\text{\textit{gd}} = 0$}

     \Comment{Base case}
     \State $s \leftarrow w_0^{u-1}(S \wedge (S \ll m))$
     \State $ng \leftarrow ng + s \cdot (r - 1)+ \binom{r}{3}$
     \If{$c > m + u$}
     \State $s \leftarrow w_0^{v-1}(S \wedge (S \ll u) \wedge (S \ll (m + u))$
     \State $ng \leftarrow ng + s$
     \EndIf
     \Else
        \State $\text{\textit{newc}} \leftarrow 0$, $s \leftarrow 2$, $\text{\textit{sold}} \leftarrow 0$, $\text{\textit{spanSg}} \leftarrow 0$
    \State $\text{\textit{gd}} \leftarrow \text{\textit{gd}} - 1$
        \State $A \gets \neg(\mathit{spanS}\ll 1)$    
        
        \Comment{Block 1: Bit 0 processing}
        \If{\Call{bit}{$S, 0$}}
            \If{\Call{bit}{$S, m$}}
                \State $ng \leftarrow \Call{Sdesc$_g$}{(S \ll 1) \vee (3 \ll (c-1)), \text{\textit{spanS}}, c + 1, m, u, v, r, \text{\textit{gd}}, \text{\textit{gcd}}}$
                \State $r \leftarrow r - 1$
            \Else
                \State $r \leftarrow r - 1$
                \State $ng \leftarrow \Call{Sdesc$_g$}{(S \ll 1)\vee (3 \ll (c-1)), \text{\textit{spanS}}, c + 1, m, u, v, r, \text{\textit{gd}}, \text{\textit{gcd}}}$
            \EndIf
        \EndIf
        
        \Comment{Block 2: Bit 1 processing}
        \If{$r > 0$ \textbf{and} \Call{bit}{$S, 1$}}
            \State $\text{\textit{gcd}} \leftarrow \Call{gcd}{c,gcd}$
            \State $\text{\textit{spanS}} \leftarrow \Call{Saddelement}{\text{\textit{spanS}}, c}$
            \If{$\text{\textit{gcd}} = 1$}
                \State $\text{\textit{spanSg}} \leftarrow \Call{genus}{\text{\textit{spanS}}}$
                \If{$\text{\textit{spanSg}} \le g $}
                    \If{$\text{\textit{spanSg}} = g$}
                        \State \Return $ng + 1$
                    \Else
                        \State \Return $ng$
                    \EndIf
                \EndIf
            \EndIf
            \State $A \leftarrow A \gg 1$
            \State $S \leftarrow S \text{ $\wedge$ } A$
            \If{\Call{bit}{$S, m + 1$}}
                \State $ng \leftarrow ng + \Call{Sdesc$_g$}{(S \ll 2) \vee (7 \gg (c-1)), \text{\textit{spanS}}, c + 2, m, u, v, r, \text{\textit{gd}}, \text{\textit{gcd}}}$
                \State $r \leftarrow r - 1$
            \Else
                \State $r \leftarrow r - 1$
                \State $ng \leftarrow ng + \Call{Sdesc$_g$}{(S \ll 2) \vee (7 \gg (c - 1)), \text{\textit{spanS}}, c + 2, m, u, v, r, \text{\textit{gd}}, \text{\textit{gcd}}}$
            \EndIf
            \State $\text{\textit{sold}} \leftarrow 1$
        \EndIf
        
    \Comment{Block 3: Larger bits processing}
        \While{$r > 0$}
            \If{\Call{bit}{$S, s$}}
                \State $\text{\textit{newc}} \leftarrow c + s + 1$
                \State $\text{\textit{spanS}} \leftarrow \Call{Saddinterval}{\text{\textit{spanS}}, c, \text{\textit{newc}} - 2}$
                \State $\text{\textit{spanSg}} \leftarrow \Call{genus}{\text{\textit{spanS}}}$
                \If{$\text{\textit{spanSg}} \le g$}
                    \If{$\text{\textit{spanSg}} = g$}
                        \State \Return $ng + 1$
                    \Else 
                        \State \Return $ng$
                    \EndIf
                \EndIf
                \For{$i \leftarrow \text{\textit{sold}}$ \textbf{to} $s - 1$}
                    \State $A \leftarrow A \gg 1$
                    \State $S \leftarrow S \wedge A$
                \EndFor
                \If{\Call{bit}{$S, m + s$}}
                    \State $ng \leftarrow ng + \Call{Sdesc$_g$}{(S \ll (s + 1)) \vee (7 \gg (\textit{newc} - 3)), \text{\textit{spanS}}, \text{\textit{newc}}, m, u, v, r, \text{\textit{gd}}, 1}$
                    \State $r \leftarrow r - 1$
                \Else
                    \State $r \leftarrow r - 1$
                    \State $ng \leftarrow ng + \Call{Sdesc$_g$}{(S \ll (s + 1)) \vee (7 \gg (\textit{newc} - 3)), \text{\textit{spanS}}, \text{\textit{newc}}, m, u, v, r, \text{\textit{gd}}, 1}$
                \EndIf
                \State $\text{\textit{sold}} \leftarrow s$
            \EndIf
            \State $s \leftarrow s + 1$
        \EndWhile
        \EndIf
    \State \Return $ng$
\EndFunction
\end{algorithmic}

\section*{Appendix B}
\begin{algorithmic}[1]
\Function{Sdesc$_F$}{$S, \mathit{spanS}, c, m, r$}
    \State $s \gets 2$, $\mathit{sold} \gets 0$
    \State $\mathit{nF} \gets 0$
   
    \Comment{Base case}
    \If{$c = F + 1$}
        \State \Return $1$
    \EndIf
    
    \State $g \gets g + 1$
    
    \Comment{Block 1: Bit 0 processing}
    \If{$\mathit{BIT}(S, 0)$}
        \If{$\mathit{BIT}(S, m)$}
            \If{$c + m = F$}
                \State \Return $2^{r - 1}$
            \Else
            \State $\mathit{nF} \gets \Call{Sdesc}{(S\ll 1) \lor (3\gg(c - 1)), \mathit{spanS}, c + 1, m, r}$
            \State $r \gets r - 1$
            \EndIf
        \Else
            \If{$c + m = F$}
                \State \Return $\mathit{nF}$
            \Else
                \State $r \gets r - 1$
                \State $\mathit{nF} \gets \Call{Sdesc$_F$}{(S\ll 1) \lor (3\gg(c - 1)), \mathit{spanS}, c + 1, m, r}$
            \EndIf
        \EndIf
    \EndIf
    \State $A \gets \neg(\mathit{spanS}\ll 1)$
    \Comment{Block 2: Bit 1 processing}
    \If{$r > 0 \land \mathit{BIT}(S, 1)$}
        \State $\mathit{spanS} \gets \Call{Saddelement}{\mathit{spanS}, c}$
        \If{$\mathit{BIT}(\mathit{spanS}, F)$}
            \State \Return $\mathit{nF}$
        \EndIf
        
        \State $A \gets (A\gg 1)$
        \State $S \gets S \land A$
        
        \If{$\mathit{BIT}(S, m + 1)$}
            \If{$c + m + 1 = F$}
            \State \Return $\mathit{nF} + 2^{r - 1}$
            \Else
                \State $\mathit{nF} \gets \mathit{nF} + \Call{Sdesc$_F$}{(S\ll 2) \lor (7\gg (c - 1)), \mathit{spanS}, c + 2, m, r}$
                \State $r \gets r - 1$
                \EndIf
        \Else
            \If{$c + m + 1 = F$}
                \State \Return $\mathit{nF}$
            \Else
                \State $r \gets r - 1$
                \State $\mathit{nF} \gets \mathit{nF} + \Call{Sdesc$_F$}{(S\ll 2) \lor (7\gg(c - 1)), \mathit{spanS}, c + 2, m, r}$
            \EndIf
        \EndIf
        \State $\mathit{sold} \gets 1$
    \EndIf
    
    \Comment{Block 3: Larger bits processing}
    \While{$r > 0$}
        \If{$\mathit{BIT}(\mathit{spanS}, F - c - s + 1)$}
            \State \Return $\mathit{nF}$
        \EndIf
        
        \If{$\mathit{BIT}(S, s)$}
            \State $\mathit{newc} \gets c + s + 1$
            \State $\mathit{spanS} \gets \Call{Saddinterval}{\mathit{spanS}, c, \mathit{newc} - 2}$
            \If{$\mathit{BIT}(\mathit{spanS}, F)$}
                \State \Return $\mathit{nF}$
            \EndIf
            
            \For{$i \gets \mathit{sold} \textrm{ to } s - 1$}
                \State $A \gets (A\gg 1)$
                \State $S \gets S \land A$
            \EndFor
            
            \If{$\mathit{BIT}(S, m + s)$}
                \If{$c + m + s = F$}
                    \State \Return $\mathit{nF} + 2^{r - 1}$
                \Else
                \State $\mathit{nF} \gets \mathit{nF} + \Call{Sdesc$_F$}{(S\ll (s + 1)) \lor (7\gg(\mathit{newc} - 3)), \mathit{spanS}, \mathit{newc}, m, r}$
                \State $r \gets r - 1$
                \EndIf
            \Else
                \If{$c + m + s = F$}
                    \State \Return $\mathit{nF}$
                \Else
                    \State $r \gets r - 1$
                    \State $\mathit{nF} \gets \mathit{nF} + \Call{Sdesc$_F$}{(S\ll(s + 1)) \lor (7\gg (\mathit{newc} - 3), \mathit{spanS}, \mathit{newc}, m, r}$
                \EndIf
            \EndIf
            \State $\mathit{sold} \gets s$
        \EndIf
        \State $s \gets s + 1$
    \EndWhile
    
    \State \Return $\mathit{nF}$
\EndFunction
\end{algorithmic}

\bibliographystyle{plain}
\bibliography{bib}
\end{document}